\documentclass{article}
\usepackage{PRIMEarxiv}
\usepackage{amsthm, amsmath, amssymb}
\usepackage{lscape}
\usepackage[pdfencoding=auto]{hyperref}

\newtheorem{theorem}{Theorem}[section]
\newtheorem{proposition}[theorem]{Proposition}
\newtheorem{corollary}[theorem]{Corollary}
\newtheorem{lemma}[theorem]{Lemma}
\newtheorem{rmk}[theorem]{Remark}

\newtheorem{example}{Example}
\usepackage[utf8]{inputenc} % allow utf-8 input
\usepackage[T1]{fontenc}    % use 8-bit T1 fonts
\usepackage{cite}
\usepackage{url}            % simple URL typesetting
\usepackage{booktabs}       % professional-quality tables
\usepackage{amsfonts}       % blackboard math symbols
\usepackage{nicefrac}       % compact symbols for 1/2, etc.
\usepackage{microtype}      % microtypography
\usepackage{lipsum}
\usepackage{fancyhdr}       % header
\usepackage{graphicx}       % graphics
\graphicspath{{media/}}     % organize your images and other figures under media/ folder

\newcommand{\biga}{\mbox{\normalfont\bfseries A}}
\newcommand{\bigd}{\mbox{\normalfont\bfseries D}}
\newcommand{\real}{{\rm I\!R}}

\title{Bounds for $A_\alpha$-eigenvalues
\thanks{ }}

\author{
  João Domingos G. da Silva Jr. \\ %\orcidlink{https://orcid.org/0000-0002-1745-0302}\\
  Departamento de Engenharia de Produção\\
  Centro Federal de Educação Tecnológica do Rio de Janeiro \\
  Rio de Janeiro, Brazil\\
  \texttt{joao.dgomes@gmail.com} \\
  \And
  Carla Silva Oliveira \\ %\orcidlink{https://orcid.org/0000-0001-6684-8811}\\
  Departamento de Matemática \\
  Escola Nacional de Ci\^encias Estat\'{\i}sticas \\
  Rio de Janeiro, Brazil\\
  \texttt{carla.oliveira@ibge.gov.br} \\
  \And
  Liliana Manuela G. C. da Costa  \\ %\orcidlink{https://orcid.org/0000-0002-5258-1447}\\
  Departamento de Matemática \\
  Col\'egio Pedro II \\
  Rio de Janeiro, Brazil\\
  \texttt{lmgccosta@gmail.com} \\
}

\begin{document}
\maketitle

\begin{abstract}
Let $G$ be a graph with adjacency matrix $A(G)$ and degree diagonal matrix $D(G)$. In 2017, Nikiforov \cite{VN17} defined the matrix $A_\alpha(G)$, as a convex combination of $A(G)$ and $D(G)$, the following way, 
$A_\alpha(G)=\alpha A(G)+(1-\alpha)D(G),$
where $\alpha\in[0,1]$. In this paper we present some new upper and lower bounds for the largest, second largest and the smallest eigenvalue of $A_\alpha$-matrix. Moreover, extremal graphs attaining some of these bounds are characterized.
\end{abstract}

% keywords can be removed

\section{Introduction} \label{sec::introduction}

Let $G=(V,E)$ be a simple graph such that $\vert V\vert = n$ and $\vert E \vert = m$. If $v_i$ is adjacent to $v_j$ we denote by $v_i \sim v_j$, otherwise $v_i \nsim v_j$. For $v \in V$, the set of its neighbours is denoted by $N_v$ and  $\vert N_v \vert$ is the cardinality of $N_v$. For each vertex $v \in V$ the degree of $v$, denoted by $d(v)$, is the number of neighbours to $v$. The minimum degree of $G$ is  $\delta(G) := \min \{ d(v); v \in V \}$ and the maximum degree of $G$ is $\Delta(G) := \max\{d(v); v \in V\}$. Eventually, we will use $d_i$, $\Delta$ and $\delta$ to represent $d(v_i)$,  $\Delta(G)$ and $\delta(G)$ respectively. Also assume that the vertices are labeled such that $\Delta = d_1 \geq d_2 \geq \ldots \geq d_n = \delta$ and the degree sequence is defined by $d(G) = (d_1, \ldots, d_n)$. A graph is said to be bidegreed if and only if its 
degree sequence has only two different values. 

The first Zagreb index is defined by $\displaystyle Z_1(G) = \sum_{i=1}^n d^2(v_i)$ and study of its bounds and properties can be found at \cite{NKMT2003,Das2003SHARPBF,DAS200457,CIOABA20061959, DasZagreb}. A nonempty subset $I \subset V$ is \textit{independent} if and only if no two of its elements are adjacent. The independence number of $G$, $\gamma(G)$, is the largest cardinality among all independent sets of $G$.

A graph is called regular if all its vertices have the same degree. The complement of the graph $G$, denoted by $\overline{G}$, is the graph  obtained from $G$ with the same vertex set, $\overline{V} = V$, and $v_iv_j \in \overline{E}$ if and only if $v_iv_j \notin E$. We denote by $K_n$, $K_{a,b}$, $K_{1,n-1}$ and $P_n$ the complete graph, the complete bipartite graph, the star and the path, respectively. The bipartite complement of a connected bipartite graph $G$ with partitions $V_1$ and $V_2$ is a bipartite graph, denoted by $\widetilde{G}$, such that $\widetilde{G}$ has edges between $V_1$ and $V_2$ exactly where $G$ does not, that is, $V(\widetilde{G}) = V(K_{\vert V_1 \vert,\vert V_2 \vert})$ and $E(\widetilde{G}) = E(K_{\vert V_1 \vert,\vert V_2 \vert}) - E(G)$ . A graph $G$ is called semi-regular bipartite, with parameters $(n_1, n_2, r_1, r_2)$, if $G$ is bipartite such that $V = V_1 \cup V_2$ where $n_1 = \vert V_1 \vert$ and $n_2 = \vert V_2 \vert$, and the vertices in the same partition have the same degree, in other words, $n_1$ vertices have degree $r_1$ and $n_2$ vertices have  degree $r_2$, such that $n_1r_1 = n_2r_2$.

Let $G = (V, E)$ and $H = (W, F)$ be disjoint graphs in their vertex sets. The coalescence between $G$ and $H$, denoted by $G \cdot H$, is a graph with $|V| + |W| - 1$ vertices that can be obtained identifying some vertex of $G$ with some vertex of $H$. Given the non-negative integers $s$ and $l$, the double kite graph, denoted by $DK(s;l)$, is the graph obtained by coalescing one of the pendent vertices of the path $P_{l+2}$ with a vertex of the complete graph $K_s$. Figure \ref{fig:dkgraph} shows the  double kite graph $D(4,3)$.

\begin{figure}[h]
\centering
\includegraphics[width=0.5\textwidth]{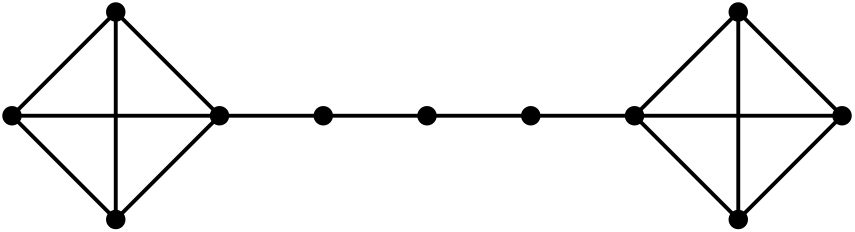}
\caption{DK(4,3).}
\label{fig:dkgraph}
\end{figure}

Let $x \in \real^n$, we denote by $\vert x \vert$ the Euclidean norm of $x$. Let $M$ be an $n \times n$ matrix. If $M$ is symmetric, the $M$-eigenvalues are real and we shall index them in non-increasing order, represented by $\lambda_1(M) \geq \ldots \geq \lambda_n(M)$. The collection of $M$-eigenvalues together with their multiplicities is called the $M$-spectrum, denoted by $\sigma(M)$.

The adjacency matrix of $G$,  $A = A(G) = [a_{ij}]$, is a square and symmetric matrix of order $n$, such that $a_{ij} = 1$ if $v_i \sim v_j$ and $a_{ij} = 0$ otherwise. The degree matrix of $G$, denoted by $D(G) = [d_{ij}]$, is the diagonal matrix such that $d_{ii} = d(v_i)$. The Laplacian  and signless Laplacian matrices are defined by $L(G) = D(G) - A(G)$ and $Q(G) = D(G) + A(G)$, respectively. An interesting problem in Graph Spectral Theory is to obtain bounds for $A$-eigenvalues, $L$-eigenvalues and $Q$-eigenvalues involving invariants associated to graphs.

In $2017$ Nikiforov, \cite{VN17}, defined for any real $\alpha \in [0,1]$ the convex linear combination, $A_\alpha(G)$, of $A(G)$ and $D(G)$ in the following way: 
$$A_\alpha(G) = \alpha D(G) + (1-\alpha)A(G) , \ \ \alpha \in [0,1].$$
It is easy to see that $A(G) = A_0(G)$, $D(G) = A_1(G)$ and $Q(G) = 2A_{\frac{1}{2}}(G)$. So, obtain bounds for $A_\alpha$-eigenvalues is an
interesting problem because it contemplates the study of bounds for the adjacency and signless Laplacian matrices.

Results involving bounds for $A_\alpha$-eigenvalues have been obtained, as we can see in \cite{LIU2020111917, LIN2018210, LIN2018430, WANG2020210, Pirzada2022, LIU2020347,Zhang2019,CHEN2019343, VN17, Abdollah2022, NIKIFOROV2017286}. In this paper we obtain some bounds for the largest eigenvalue, the second largest eigenvalue and the smallest eigenvalue of the $A_\alpha$ matrix. This paper is organized as follows: in Section \ref{sec::preliminaries} we introduce some definitions and results required to prove the main results; after, in Section \ref{sec::main}  we show the main results referring to bounds for $A_\alpha$-eigenvalues.

\section{Preliminaries} \label{sec::preliminaries}

In this section we present some aspects of matrix theory that will be needed to prove the main results of this paper. The principal sub-matrix of a matrix is obtained by removing rows and columns with the same indices \cite{Goldberg59}. An important result about sub-matrices is presented in Theorem \ref{theo::interlacing}.

\begin{theorem} \label{theo::interlacing}
\cite{horn2013matrix} Suppose $A \in M_n({\rm I\!R})$ symmetric with eigenvalues $\lambda_1 \geq \ldots \geq \lambda_n$. If $B \in M_m({\rm I\!R})$ with $ m<n$, a principal sub-matrix of $A$ with eigenvalues $\mu_1 \geq \ldots \geq \mu_m$, then $\lambda_{n-m+i} \leq \mu_i \leq \lambda_i$, for $i = 1, \ldots, m$.
\end{theorem}

The following result is the theorem of Weyl and So, which is inequalities involving eigenvalues of sums of Hermitian matrices. 

\begin{theorem}[Weyl] \label{theo::weyl}
\cite{horn2013matrix} Let $A, B \in M_n({\rm I\!R})$ be Hermitian and let the spectrum of $A$, $B$, and $A + B$ be $\sigma(A) = \{\lambda_1(A), \ldots, \lambda_n(A)\}$, $\sigma(B) = \{\lambda_1(B), \ldots, \lambda_n(B)\}$ and $\sigma(A+B) = \{\lambda_1(A+B), \ldots, \lambda_n(A+B)\}$, respectively. Then,
\begin{equation}\label{weyl_upper_bound}
\lambda_{i+j-1}(A+B) \leq \lambda_i(A) + \lambda_j(B), \ \ j = 1, \ldots , n-i+1
\end{equation}
for each $i = 1, \ldots, n$, with equality for some pair $i,j$ if and only if there is a nonzero vector $x$ such that $Ax = \lambda_i x$, $Bx = \lambda_j x$ and $(A+B)x = \lambda_{i+j-1} x$.
Also,	\begin{equation}\label{weyl_lower_bound}
\lambda_i(A) + \lambda_j(B) \leq \lambda_{i + j -n}(A+B), \ \ j = i, \ldots, n
\end{equation}
for each $i = 1, \ldots, n$, with equality for some pair $i,j$ if and only if there is a nonzero vector $x$ such that $Ax = \lambda_i x$, $Bx = \lambda_j x$ e $(A+B)x = \lambda_{i+j-n} x$. If $A$ and $B$ have no common eigenvector, then the inequalities in (\ref{weyl_upper_bound}) and (\ref{weyl_lower_bound}) are strict.
\end{theorem}

As consequence of Theorem \ref{theo::weyl}, follows Corollary \ref{cor::weyl}.

\begin{corollary} \label{cor::weyl}
\cite{horn2013matrix} Let be $A, B \in M_n({\rm I\!R})$ Hermitian. Then,
\begin{equation} \label{weyl_inequalities}
\lambda_i(A) + \lambda_n(B) \leq \lambda_i(A+B) \leq \lambda_i(A) + \lambda_1(B),
\end{equation}
with $i = 1, \ldots, n$. Equality in the upper bound holds if and only if there is nonzero vector $x$  that is eigenvector of $A,B$ and $ A+B$ with corresponding eigenvalues $\lambda_i$, $\lambda_1$ and $\lambda_i$, respectively. Analogously, equality in the lower bound holds if and only if there is nonzero vector $x$  that is eigenvector of $A,B$ and $ A+B$ with corresponding eigenvalues $\lambda_i$, $\lambda_n$ and $\lambda_i$, respectively.
\end{corollary}

\begin{lemma} \label{lemma::rowsums}
\cite{ELLINGHAM200045} Let $G$ be a connected graph with $n$ vertices and $A(G)$ its adjacency matrix. Let $P(x)$ be any polynomial function and $S_v(P(A(G)))$ be the row sums of $P(A(G))$ corresponding to each vertex $v$. Then
$$\min{S_v(P(A))} \leq P(\lambda_1(A(G))) \leq \max{S_v(P(A))}.$$
Moreover, equality holds if and only if the row sums of $P(A(G))$ are all equal.
\end{lemma}

The proof of Lemma \ref{lemma:: rowsumsBound} is presented because it is used to prove Theorem \ref{theo::UpperBound2_SpectralRadius} in the next section.

\begin{lemma} \label{lemma:: rowsumsBound}
\cite{HONG2001177} Let $G$ be a graph with $n$ vertices, $m$ edges and minimum degree $\delta$. Then, $S_v(A^2(G) - (\delta - 1)A(G)) \leq 2m -\delta(n-1)$.
\end{lemma}
\begin{proof}
Note that $S_v(A^k(G))$ is exactly the number of walks of length $k$ in $G$ which begin at $v$. In particular, $S_v(A(G))$ is $d(v)$ and $S_v(A^2(G))= \displaystyle \sum_{u \sim v} d(u)$. So,
\begin{align*}
    S_v(A^2(G)) &= \displaystyle \sum_{u \sim v} d(u) \\
    &= 2m - d(v) - \displaystyle \sum_{u \nsim v, u \neq v} d(u)\\
    &\leq 2m - d(v) - (n - d(v) - 1)\delta \\
    &= 2m + (\delta - 1)d(v) - \delta(n-1).
\end{align*}
Hence,
$$S_v(A^2(G) - (\delta - 1)A(G)) \leq 2m - \delta(n-1).$$
\end{proof}

The next two theorems present a lower and an upper bound, respectively, for $Z_1(G)$ using $\Delta$, $\delta$, $m$ and $n$.

\begin{theorem} \label{theo::sharp_DAS}
	\cite{Das2003SHARPBF} Let $G$ be a simple graph with $n$ vertices and $m$ edges. Let $\delta$ and $\Delta$ be the minimum  and the maximum degree of $G$, respectively. Then, for $n \geq 3$, $\displaystyle Z_1(G) \geq \Delta^2 + \delta^2 + \frac{(2m - \Delta - \delta)^2}{n-2}$. Furthermore, equality occurs if, and only if, $d_2 = \ldots = d_{n-1}$.
\end{theorem}

\begin{theorem} \label{theo::sharp_DAS_upper}
\cite{Das2003SHARPBF} Let $G$ be a connected graph with $n$ vertices and $m$ edges. Let $\delta$ be the minimum degree of $G$.  Then, $\displaystyle Z_1(G) \leq 2mn -n(n-1)\delta + 2m(\delta-1)$. Moreover, the equality holds if, and only if, $G$ is a star graph or a regular graph.
\end{theorem}
The next results involve properties and bounds for the largest, the second largest and the smallest eigenvalues of $A(G)$.

\begin{theorem} \label{theo::bipartite_spectrum}
\cite{cvetkovic2009} A graph $G$ is bipartite if and only if its spectrum is symmetric about the origin. 
\end{theorem}

\begin{proposition} \label{prop::regular_graph_spectrum}
\cite{cvetkovic2009} Let $G$ be a $r$-regular graph. Then
\begin{itemize}
\item [(i)] $r$ is an eigenvalue of $A(G)$;
\item[(ii)] $G$ is a connected graph if and only if the algebraic multiplicity of $r$ is $1$;
\item[(iii)] any $\lambda$ eigenvalue of $A(G)$ satisfies $|\lambda| \leq r$.
\end{itemize}
\end{proposition}

\begin{proposition} \label{prop::LowerBound}
\cite{HOFFMAN1995105} Let $G$ be a graph with $m$ edges, then
$$\lambda_1(A(G)) \geq \frac{1}{m}\sum_{i \sim j}\sqrt{d_id_j}.$$
Equality holds if, and only if, $G$ is regular or semi-regular bipartite.
\end{proposition}

\begin{theorem} \label{theo::UpperBound_Hong}
\cite{YUAN1988135} Let $G$ be a connected graph with $m$ edges and $n$ vertices. Then,
$$\lambda_1(A(G)) \leq \sqrt{2m-n+1},$$
with equality if, and only if, $G \cong K_n$ or $G \cong K_{1, n-1}$.
\end{theorem}

\begin{theorem}\label{theo::UpperBound_Pintu}
\cite{Pintu2019} Suppose $G$ be graph with $n$ vertices and $m$ edges. Let $\lambda_1(A(G))$
be the largest eigenvalue of the adjacency matrix $A(G)$. Then
\begin{equation*}
    \lambda_1(A(G)) \leq \displaystyle \sqrt{\max_{1 \leq i \leq n} \sum_{\substack{j \\ v_i \sim v_j}}d_j}
\end{equation*}
\end{theorem}

\begin{theorem} \label{theo::LowerBound_SecondLargest_adjacency}
\cite{NIKIFOROV2006612} Let $G$ be a $r$-regular graph of order $n$ and independence number $\gamma(G)$, then 
$$ \lambda_2(A(G)) \geq -1 + \dfrac{2(n-1-r)^{\gamma(G)}}{\gamma(G) n^{\gamma(G) - 1}} $$
\end{theorem}

\begin{theorem}  \label{theo::lambda2_KS_bound}
	\cite{KS2013} Let $G$ be a $r$-regular bipartite connected graph with $n$ vertices. Then $\displaystyle \lambda_2(A(G)) \leq \frac{n}{2} - r$.
\end{theorem}

\begin{corollary}\label{cor::KS2013_corollary31}
	\cite{KS2013} Let $G$ be a $r$-regular bipartite graph with $n$ vertices. Then $\lambda_1(A(G)) + \lambda_2(A(G)) \leq \displaystyle \frac{n}{2}$. Furthermore, $\displaystyle \lambda_1(A(G)) + \lambda_2(A(G)) = \frac{n}{2}$ if and only if its bipartite complement is disjoint.
\end{corollary}

\begin{theorem} \label{lambda_nGeneralBound}
    \cite{Haemers1980} Let $G$ be a graph with minimum degree $\delta \neq 0$ and independence number $\gamma(G)$, then
    \begin{equation*}
        \lambda_n(A(G)) \leq \dfrac{\gamma(G)\delta^2}{\lambda_1(A(G))(\gamma(G) - n)}
    \end{equation*}
\end{theorem}

\begin{theorem} \label{theo::HoffmanBound}
    \cite{cvetkovic2009,HAEMERS2021215} If $G$ is a $r$-regular graph with independence number $\gamma(G)$, then
    \begin{equation*}
        \lambda_n(A(G)) \leq \dfrac{\gamma(G)r}{\gamma(G) - n}
    \end{equation*}
\end{theorem}

\begin{lemma} \label{lemma::lambda_nBoundTriangleFree}
    \cite{CSIKVARI202292} Let $G$ be a triangle-free graph on $n$ vertices. Then,
    \begin{equation*}
        \lambda_n(A(G)) \leq \dfrac{\lambda_1^2(A(G))}{\lambda_1(A(G)) -n}.
    \end{equation*}
\end{lemma}

\begin{theorem} \label{theo::equality_roots_polynomial}
	\cite{KS2013} Let $G$ be a $r$-regular bipartite connected graph with $2n$ vertices and $\widetilde{G}$ its bipartite complement. Then,
	$$\displaystyle \frac{P_{A(G)}(\lambda)}{\lambda^2 - r^2} = \frac{P_{A(\widetilde{G})}(\lambda)}{\lambda^2 - (n-r)^2}.$$
\end{theorem}

\begin{theorem} \label{theo::LargestEigenvaluesSum}
\cite{EBRAHIMIB20082781} If $G$ is a regular graph of order $n$, then $\lambda_1(A(G)) + \lambda_2(A(G)) \leq n - 2$. Moreover, $\lambda_1(A(G)) + \lambda_2(A(G)) = n - 2$ if, and only if, the complement of $G$ has a component that is a bipartite graph.
\end{theorem}

\begin{proposition} \label{theo::upper_bound_laplacian_eigenvalue}
\cite{cvetkovic2009} Let $\sigma(L(G)) = \{\mu_1, \ldots ,\mu_n \}$ be the $L(G)$-spectrum such that $\mu_1 \geq \ldots \geq \mu_n$. If $G$ is a graph with $n$ vertices then $\mu_1 \leq n$, with equality occurring if and only if $\overline{G}$ is disconnected. 
\end{proposition}

The next results refer to the $A_\alpha$-matrix.

\begin{proposition} \label{rayleigh_alpha}
\cite{VN17} If $\alpha \in [0,1]$ and $G$ is a graph of order $n$, then
\begin{equation}
\lambda_1(A_\alpha(G)) = \max_{\vert x \vert = 1} \langle A_\alpha(G) x, x \rangle \text{ and }  \lambda_n(A_\alpha(G)) = \min_{\lvert x \rvert = 1} \langle A_\alpha(G) x, x \rangle.
\end{equation}
Furthermore, if $x$ is a unit vector, then $\lambda_1(A_\alpha(G)) = \langle A_\alpha(G) x, x \rangle $ if, and only if, $x$ is an eigenvector of $\lambda_1(A_\alpha(G))$, and $\lambda_n(A_\alpha(G)) = \langle A_\alpha(G) x, x \rangle $ if, and only if, $x$ is an eigenvector of $\lambda_n(A_\alpha(G))$.
\end{proposition}

\begin{lemma} \label{lemma::eigeneq_RegularGraphs}
\cite{VN17} If $\alpha \in [0,1]$ and $k = 1, \ldots, n$ and $G$ is a $r$-regular graph of order $n$, then there exists a linear correspondence between the eigenvalues of $A_\alpha(G)$ and $A(G)$, the following way
\begin{equation} \label{eq::autoequation}
\lambda_k(A_\alpha(G)) = \alpha r + (1-\alpha)\lambda_k(A(G)).
\end{equation}
In particular, if $G$ is $r$-regular, then $\lambda_1(A_\alpha(G)) = r, \ \ \forall \alpha \in [0,1]$.
\end{lemma}

\begin{proposition} \label{prop::Perron_alpha}
\cite{VN17} Let $\alpha \in [0,1)$, $G$ be a graph and $x$ be a nonnegative eigenvector of $\lambda_1(A_\alpha(G))$.
\begin{itemize}
\item [(i)] If $G$ is connected, then $x$ is positive and unique minus scalar;
\item [(ii)] If $G$ is disconnected and $P$ is the set of vertices with positive entries of $x$, then the subgraph induced by $P$ is a union of $H$ components of $G$ with $\lambda_1(A_\alpha(H)) = \lambda_1(A_\alpha(G))$;
\item [(iii)] If $G$ is connected and $\mu$ is an eigenvalue of $A_\alpha(G)$ with a non-negative eigenvector, then $\mu = \lambda_1(A_\alpha(G))$;
\item [(iv)] If $G$ is connected, and $H$ is an eigengraph subgraph of $G$, then $\lambda_1(A_\alpha(H))< \lambda_1(A_\alpha(G))$.
\end{itemize}
\end{proposition}

\begin{proposition} \label{prop::complete_graph_spectrum}
\cite{VN17} The eigenvalues of $A_\alpha(K_n)$ are $\lambda_1(A_\alpha(K_n)) = n-1$ and $\lambda_k(A_\alpha(K_n)) = \alpha n -1 \text{  for } 2 \leq k \leq n$.
\end{proposition}

\begin{proposition} \label{prop::complete_bipartite_spectrum}
\cite{VN17} Let $a \geq b \geq 1$. If $\alpha \in [0,1]$, the eigenvalues of $A_\alpha(K_{a,b})$ are
\begin{align*}
    \lambda_1(A_\alpha(K_{a,b})) &= \dfrac{1}{2}\left(\alpha(a+b) + \sqrt{\alpha^2(a+b)^2 + 4ab(1-2\alpha)} \right),\\
    \lambda_{\min}(A_\alpha(K_{a,b})) &= \dfrac{1}{2}\left(\alpha(a+b) - \sqrt{\alpha^2(a+b)^2 + 4ab(1-2\alpha)} \right),\\
    \lambda_k(A_\alpha(K_{a,b})) &= \alpha a \text{ for } 1 < k \leq b,\\
    \lambda_k(A_\alpha(K_{a,b})) &= \alpha b \text{ for } b < k < a+b.
\end{align*}
\end{proposition}

\section{Main Results} \label{sec::main}

This section presents the main results of this paper which involve bounds for the largest, the second largest and the smallest eigenvalues of $A_\alpha$-matrix.

\subsection{Bounds for \texorpdfstring{$\mathbf{\lambda_1(A_\alpha(G))}$}{lambda1}} \label{subsection::lambda1}

\begin{theorem} \label{theo::LowerBound_SpectralRadius}
Let $\alpha \in [0,1]$ and $G$ be a graph with $m \neq 0$ edges, $n \geq 3$ vertices, $\Delta$ and $\delta$ the maximum and minimum degrees, respectively. Then,
\begin{equation} \label{eq::lowerbound1}
    \displaystyle \lambda_1(A_\alpha(G)) \geq \left( \Delta^2 + \delta^2 + \frac{(2m -\Delta - \delta)^2}{n-2} \right)\frac{\alpha}{2m} + \frac{1-\alpha}{m} \sum_{i \sim j}\sqrt{d_id_j}.
\end{equation}
The equality occurs if and only if $G$ is a regular graph or $\displaystyle G \cong \bigcup_{k=1}^{l-1} G_k \cup K_1$ such that $G_k$ is $r$-regular.
\end{theorem}
\begin{proof}
From Proposition \ref{rayleigh_alpha} we know that exist an eigenvector $x \in \real^n$ associated to $\lambda_1(A_\alpha(G))$ that satisfies $\displaystyle \lambda_1(A_\alpha(G)) = \max_{x \in \real^n} \frac{x^T A_\alpha(G) x}{x^Tx}$. So  for all $y \neq kx$,  $\; k \in \real$, we have
$$\lambda_1(A_\alpha(G)) \geq  \frac{y^T A_\alpha(G) y}{y^Ty} = \frac{y^T(\alpha D(G) + (1-\alpha)A(G)) y}{y^Ty}=\frac{\alpha y^T D(G)y + (1-\alpha)y^TA(G) y}{y^Ty}.$$
Taking $y = (\sqrt{d_1}, \sqrt{d_2}, \ldots, \sqrt{d_n})$, follows that
{\small
$$\lambda_1(A_\alpha(G)) \geq \dfrac{\alpha \displaystyle \sum_{i=1}^n d_{i}^2+ (1-\alpha)y^TA(G)) y}{2m}=\dfrac{\alpha \displaystyle \sum_{i=1}^n d_{i}^2}{2m} + \dfrac{(1-\alpha) \displaystyle \sum_{v_i \sim v_j} 2\sqrt{d_{i}d_{j}}}{2m}= \dfrac{\alpha \displaystyle \sum_{i=1}^n d_{i}^2}{2m} + \dfrac{(1-\alpha)}{m} \displaystyle \sum_{v_i \sim v_j} \sqrt{d_{i}d_{j}}
$$}
From Theorem \ref{theo::sharp_DAS} follows that
\begin{equation} \label{ineq::SpactralRadiusBound_alpha}
	\lambda_1(A_\alpha(G)) \geq \dfrac{\alpha}{2m} \left[ \Delta^2 + \delta^2 + \dfrac{(2m - \Delta - \delta)^2}{n-2}\right] + \dfrac{(1-\alpha)}{m} \sum_{v_i \sim v_j} \sqrt{d_{i}d_{j}} 
\end{equation}

Suppose initially that $G$ is $r$-regular graph. So $\Delta = \delta = r$ and $m = \dfrac{nr}{2}$. Then,
\begin{align*}
	&\dfrac{\alpha}{2m} \left[ \Delta^2 + \delta^2 + \dfrac{(2m - \Delta - \delta)^2}{n-2}\right] + \dfrac{(1-\alpha)}{m} \sum_{v_i \sim v_j} \sqrt{d_{i}d_{j}} =
	\dfrac{\alpha}{nr} \left[ r^2 + r^2 + \dfrac{(nr - r - r)^2}{n-2}\right] + \dfrac{(1-\alpha)}{m} mr = \\
	&\dfrac{\alpha}{nr} \left[ 2r^2 + \dfrac{r^2(n -2)^2}{n-2} \right] + (1-\alpha)r =
	\dfrac{\alpha}{nr} r^2\left[ 2 + n -2 \right] + (1-\alpha)r =
	 \alpha r + (1-\alpha)r = r
\end{align*}

Now, suppose that $\displaystyle G \cong \bigcup_{k=1}^{l-1} G_k \cup K_1$ such that $G_k$ is $r$-regular. We have $\displaystyle \sum_{k=1}^{l-1}|V(G_k)| + 1 = n,$  $\Delta = r$, $\delta = 0$ and $m = \dfrac{(n-1)r}{2}$. So, 
\begin{align*}
	&\dfrac{\alpha}{2m} \left[ \Delta^2 + \delta^2 + \dfrac{(2m - \Delta - \delta)^2}{n-2}\right] + \dfrac{(1-\alpha)}{m} \sum_{v_i \sim v_j} \sqrt{d_{i}d_{j}} = 
	\dfrac{\alpha}{(n-1)r} \left[ r^2  + \dfrac{((n-1)r - r)^2}{n-2}\right] + \dfrac{(1-\alpha)}{m} mr = \\
	&\dfrac{\alpha}{(n-1)r} \left[ r^2 + \dfrac{r^2(n -2)^2}{n-2} \right] + (1-\alpha)r =
	\dfrac{\alpha}{(n-1)r} r^2\left[ n -1 \right] + (1-\alpha)r =
	 \alpha r + (1-\alpha)r = r.
\end{align*}
Moreover, from Lemma \ref{lemma::eigeneq_RegularGraphs} we have $\lambda_1(A_\alpha(G)) = r$.

Now, suppose there is a graph $G$ that satisfies the equality
$$\displaystyle \lambda_1(A_\alpha(G)) = \left( \Delta^2 + \delta^2 + \frac{(2m -\Delta - \delta)^2}{n-2} \right)\frac{\alpha}{2m} + \frac{1-\alpha}{m} \sum_{v_i \sim v_j}\sqrt{d_id_j}.$$
This implies that $x = (\sqrt{d_{1}}, \ldots,\sqrt{d_{n}})$ is eigenvector of $A_\alpha(G)$ associated to $\lambda_1(A_\alpha(G))$, this is $A_\alpha(G)x = \lambda_1(A_\alpha(G))x$.

If $G$ is connected, we have \begin{equation*}
    \alpha d_i + (1-\alpha) \dfrac{\displaystyle \sum_{v_j \sim v_i} \sqrt{d_id_j}}{d_i}  = \lambda_1(A_\alpha(G)),
\end{equation*} 
for all $i = 1, \ldots n$. So, for arbitrary $i$ and $s$, such that $i \neq s$  we have
\begin{equation*}
    \alpha d_i + (1-\alpha) \dfrac{\displaystyle \sum_{v_j \sim v_i} \sqrt{d_id_j}}{d_i}  = \alpha d_s + (1-\alpha) \dfrac{\displaystyle \sum_{v_j \sim v_s} \sqrt{d_sd_j}}{d_s}  
\end{equation*}
which implies that $d_i = d_s$ and then $G$ is regular.

Now, suppose that $G$ is disconnected. Then $ \displaystyle G \cong \bigcup_{k=1}^{l} G_k$ and consequently for each component such that $|E(G_k)| \neq 0$, we also have
\begin{equation*}
    \alpha d_i + (1-\alpha) \dfrac{\displaystyle \sum_{v_j \sim v_i} \sqrt{d_id_j}}{d_i}  = \lambda_1(A_\alpha(G_k)),
\end{equation*}
for all $i = 1, \ldots ,\vert V(G_k) \vert$. So, for arbitrary $i$ and $s$, such that $i \neq s$  we have
\begin{equation*}
    \alpha d_i + (1-\alpha) \dfrac{\displaystyle \sum_{v_j \sim v_i} \sqrt{d_id_j}}{d_i}  = \alpha d_s + (1-\alpha) \dfrac{\displaystyle \sum_{v_j \sim v_s} \sqrt{d_sd_j}}{d_s}  
\end{equation*}
which implies that $d_i = d_s$ and then $G_k,$ $\forall k,$ is regular. From Theorem \ref{theo::sharp_DAS} $\displaystyle Z_1(G) = \Delta^2 + \delta^2 + \frac{(2m - \Delta - \delta)^2}{n-2}$ must occurs if and only if $d_2 = \ldots = d_{n-1}$. Then, $\displaystyle G \cong \bigcup_{k=1}^{l} G_k$ or $\displaystyle G \cong \bigcup_{k=1}^{l-1} G_k \cup K_1$ where $G_k$ is $r$-regular.
\end{proof}

Theorem \ref{theo::LowerBound_SpectralRadius2} presents other lower bound of $\lambda_1(A_\alpha(G))$. This bound is obtained similarly of the lower bound for $\lambda_1(A(G))$, obtained by Kumar, \cite{Kumar2010}. Denote by $(\biga_\alpha)_i$, $\biga_i$ and $\bigd_i$ to represent the $i^{th}$ column vector of the matrices $A_\alpha(G)$, $A(G)$ and $D(G)$, respectively.  Furthermore, we define $c_{ij} = \vert N_{v_i} \cap N_{v_j} \vert$.

\begin{theorem} \label{theo::LowerBound_SpectralRadius2}
   Let $G$ be a graph with $n \geq 2$ vertices and $\alpha \in [0,1]$. Then
   {\small
   \begin{equation} \label{eq::lowerbound2}
       \lambda_1(A_\alpha(G)) \geq \max \left\{ \max_{j < i}\sqrt{\dfrac{\alpha^2(d_i^2 + d_j^2) + (1-\alpha)^2(d_i+d_j) + \sqrt{C}}{2}},\max_{j < i} \sqrt{\dfrac{\alpha^2(d_i^2 + d_j^2) + (1-\alpha)^2(d_i+d_j) + \sqrt{F}}{2}} \right\},
   \end{equation}
   }
where $C = \alpha^4(d_i^4 + d_j^4) + 2\alpha^2(1-\alpha)^2(d_i^3+d_j^3) + 8\alpha c_{ij}d_i(1-\alpha)^3 - 2d_j^2d_i\alpha^2(\alpha^2d_i + (1-\alpha)^2) - 2d_j\alpha(1-\alpha)^2(\alpha d_i^2 + 2c_{ij}(\alpha - 1)) +(4c_{ij}^2 + d_i^2 + d_j^2)(1-\alpha)^4 + 2d_id_j(7\alpha^2 + 2\alpha -1)(1-\alpha)^2$ and $F = (d_i^2 - d_j^2)^2 + 2\alpha^2(1-\alpha)^2(d_i^3 + d_i^2d_j - d_id_j^2 + d_j^3) + (1-\alpha)^4(4c_{ij} + (d_i + d_j)^2)$
\end{theorem}
\begin{proof}
We know that $\lambda_1(M)$, for all  symmetric matrix $M$, is greater than or equal to the largest eigenvalue of any principal sub-matrix of $M$. We also know that  any principal sub-matrix of order two of $A_\alpha^2(G)$ is of the form
$B = \begin{bmatrix}
(\biga_\alpha^T)_i(\biga_\alpha)_i & (\biga_\alpha^T)_i(\biga_\alpha)_j \\
(\biga_\alpha^T)_j(\biga_\alpha)_i & (\biga_\alpha^T)_j(\biga_\alpha)_j
\end{bmatrix}.$
Then,
{\small
\begin{align}
&\begin{bmatrix}
	(\biga_\alpha^T)_i(\biga_\alpha)_i & (\biga_\alpha^T)_i(\biga_\alpha)_j \\
	(\biga_\alpha^T)_j(\biga_\alpha)_i & (\biga_\alpha^T)_j(\biga_\alpha)_j
\end{bmatrix} =
\begin{bmatrix}
	(\alpha \bigd_i^T + (1-\alpha)\biga_i^T)(\alpha \bigd_i + (1-\alpha)\biga_i) & (\alpha \bigd_i^T + (1-\alpha)\biga_i^T)(\alpha \bigd_j + (1-\alpha)\biga_j) \\
	(\alpha \bigd_j^T + (1-\alpha)\biga_j^T)(\alpha \bigd_i + (1-\alpha)\biga_i) & (\alpha \bigd_j^T + (1-\alpha)\biga_j^T)(\alpha \bigd_j + (1-\alpha)\biga_j)
\end{bmatrix} = \nonumber \\
& \begin{bmatrix}
    \alpha^2 \bigd_i^T\bigd_i + \alpha(1-\alpha)(\bigd_i^T\biga_i + \biga_i^T\bigd_i) + (1-\alpha)^2\biga_i^T\biga_i & \alpha^2 \bigd_i^T\bigd_j + \alpha(1-\alpha)(\bigd_i^T\biga_j + \biga_i^T\bigd_j) + (1-\alpha)^2\biga_i^T\biga_j \\
    \alpha^2\bigd_j^T\bigd_i + \alpha(1-\alpha)(\bigd_j^T\biga_i + \biga_j^T\bigd_i) + (1-\alpha)^2\biga_j^T\biga_i & \alpha^2 \bigd_j^T\bigd_j + \alpha(1-\alpha)(\bigd_j^T\biga_j + \biga_j^T\bigd_j) + (1-\alpha)^2\biga_j^T\biga_j
\end{bmatrix} \label{eq::powMatrix}
\end{align}
}
We have two cases to consider. The first one is if $v_i \sim v_j$, from (\ref{eq::powMatrix}) we have
\begin{align*}
\begin{bmatrix}
	(\biga_\alpha^T)_i(\biga_\alpha)_i & (\biga_\alpha^T)_i(\biga_\alpha)_j \\
	(\biga_\alpha^T)_j(\biga_\alpha)_i & (\biga_\alpha^T)_j(\biga_\alpha)_j
\end{bmatrix} =
\begin{bmatrix}
    \alpha^2d_i^2 + (1-\alpha)^2d_i & \alpha(1-\alpha)(d_i+d_j) + (1-\alpha)^2c_{ij}\\
    \alpha(1-\alpha)(d_i+d_j) + (1-\alpha)^2c_{ij} & \alpha^2d_j^2 + (1-\alpha)^2d_j
\end{bmatrix}.
\end{align*}
As $\lambda_1(A_\alpha^2(G)) \geq \lambda_1(B)$ we obtain 

\begin{equation*}
    \lambda_1(A_\alpha(G)) \geq \dfrac{1}{\sqrt{2}}\sqrt{\alpha^2(d_i^2 + d_j^2) + (1-\alpha)^2(d_i+d_j) + \sqrt{C}}
\end{equation*}
where $C = \alpha^4(d_i^2 - d_j^2)^2 + 2\alpha^2(1-\alpha)^2(d_i^3+d_j^3) -8\alpha c_{ij}(\alpha - 1)^3(d_i + d_j) + (\alpha-1)^2(5\alpha^2 -2\alpha+1)(d_i^2 + d_j^2) - 2\alpha^2(1-\alpha)^2d_id_j(d_i+d_j) + 4c_{ij}^2(\alpha - 1)^4 + 2d_id_j(1-\alpha)^2(1+\alpha)(3\alpha - 1)$ 

The second one is if $v_i \nsim v_j$. In this case, from (\ref{eq::powMatrix}) we have  that
\begin{align*}
B = \begin{bmatrix}
	(\biga_\alpha^T)_i(\biga_\alpha)_i & (\biga_\alpha^T)_i(\biga_\alpha)_j \\
	(\biga_\alpha^T)_j(\biga_\alpha)_i & (\biga_\alpha^T)_j(\biga_\alpha)_j
\end{bmatrix} =
\begin{bmatrix}
    \alpha^2d_i^2 + (1-\alpha)^2d_i & (1-\alpha)^2c_{ij}\\
    (1-\alpha)^2c_{ij} & \alpha^2d_j^2 + (1-\alpha)^2d_j
\end{bmatrix}.
\end{align*}

Then, 
\begin{equation*}
    \lambda_1(A_\alpha(G)) \geq \dfrac{1}{\sqrt{2}}\sqrt{\alpha^2(d_i^2 + d_j^2) + (1-\alpha)^2(d_i+d_j) + \sqrt{F}}
\end{equation*}
where $F = (d_i^2 - d_j^2)^2 + 2\alpha^2(1-\alpha)^2(d_i^3 + d_i^2d_j - d_id_j^2 + d_j^3) + (1-\alpha)^4(4c_{ij} + (d_i + d_j)^2)$ and  the result follows.
\end{proof}
 
\begin{example} Let $G$ be a graph in Figure \ref{fig::graph1}. Table \ref{tab:table1} shows that the lower bounds presented in $(\ref{eq::lowerbound1})$ and $(\ref{eq::lowerbound2})$ are incomparable.
\begin{figure}[!h]%
    \centering
    \includegraphics[width=0.2\textwidth]{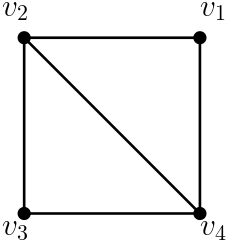}
    \centering
    \caption{Graph $G$.}
    \label{fig::graph1}
\end{figure}
\begin{table}[h]
    \centering
    \resizebox{16cm}{!}
    {
    \begin{tabular}{|c|c|c|c|c|c|c|c|c|c|c|}
    \hline
    & $\alpha = 0.0$ & $\alpha = 0.1$ & $\alpha = 0.2$ & $\alpha = 0.3$ & $\alpha = 0.4$ & $\alpha = 0.5$ & $\alpha = 0.6$ & $\alpha = 0.7$ & $\alpha = 0.8$ & $\alpha = 0.9$ \\  
    \hline
    $\lambda_1(A_\alpha(G))$ & 2.56155 &  2.56815 & 2.57631 & 2.58661 & 2.6 & 2.61803 & 2.6434 & 2.68102 & 2.74031 & 2.83852\\
    \hline
    (\ref{eq::lowerbound1}) & 2.55959 & 2.55863 & 2.55767 & 2.55671 & 2.55576 & 2.5548 & 2.55384 & 2.55288 & 2.55192 & 2.55096\\ 
    \hline
    (\ref{eq::lowerbound2}) & 2.23607 & 2.16333 & 2.12603 & 2.12603 & 2.16333 & 2.23607 & 2.34094 & 2.47386 & 2.63059 & 2.80713\\
    \hline 
    \end{tabular}
    }
    \caption{Table with lower bounds of $\lambda_1(A_\alpha(G))$ with different $\alpha$ values.}
    \label{tab:table1}
\end{table}

\end{example}

\begin{theorem}
If $G$ is a connected graph with $n$ vertices, $m$ edges, maximum degree $\Delta$, minimum degree $\delta$ and $\alpha \in [0,1]$, then 
\begin{equation} \label{eq::UpperBound1}
    \lambda_1(A_\alpha(G)) \leq \sqrt{\alpha^2 (2mn -n(n-1)\delta + 2m(\delta-1)) + 2m(1-\alpha)^2 + \delta(\alpha^2\delta+(1-\alpha)^2)(\Delta - n + 1)}.
\end{equation}
\end{theorem}
\begin{proof}
Denote by $(A_\alpha)_i$ the $i$-th row of the matrix $A_\alpha(G)$. From Proposition \ref{prop::Perron_alpha}, let $X = (x_1, \ldots, x_n)$ be the unit positive eigenvector associated to $\lambda_1(A_\alpha(G))$. Denote $X^i$ the vector obtained from $X$ by replacing $x_j$ by $0$ if $v_j$ is not adjacent to $v_i$. Since $A_\alpha(G) X = \lambda_1(A_\alpha(G))X$, we have that 
$\lambda_1(A_\alpha(G)) x_i = (A_\alpha)_iX = (A_\alpha)_iX^i.$

From the Cauchy-Schwarz inequality, we have
\begin{equation}\label{ineq::cauchy}
	\lambda_1^2(A_\alpha(G)) x_i^2 = |(A_\alpha)_iX^i|^2 \leq |(A_\alpha)_i|^2|X^i|^2 = \left(d_i^2\alpha^2 + (1-\alpha)^2d_i \right)\left(1-\sum_{v_i \nsim v_j}x_j^2 \right).
\end{equation}
Taking the inequality (\ref{ineq::cauchy}) for all $i,$ we have
\begin{equation} \label{ineq1::UpperBound_radius} 
	\lambda_1^2(A_\alpha(G)) \leq \alpha^2 \sum_{i=1}^n d_i^2 + (1-\alpha)^2\sum_{i=1}^n d_i - \sum_{i=1}^n \left(d_i^2\alpha^2 + (1-\alpha)^2d_i\right)\left(\sum_{v_i \nsim v_j} x_j^2 \right).
\end{equation}
Since
\begin{align}
	\sum_{i=1}^n \left(d_i^2\alpha^2 + (1-\alpha)^2d_i\right)\left(\sum_{v_i \nsim v_j} x_j^2 \right) &= \alpha^2 \sum_{i=1}^nd_i^2\left(\sum_{\substack{v_i \nsim v_j}}x_j^2 \right) 
	+ (1-\alpha)^2 \sum_{i=1}^n d_i\left(\sum_{\substack{v_i \nsim v_j}}x_j^2 \right) \nonumber\\
	&\geq \alpha^2 \delta^2 \sum_{i=1}^n \left(\sum_{\substack{v_i \nsim v_j}}x_j^2 \right) 
	+ (1-\alpha)^2 \delta \sum_{i=1}^n\left(\sum_{\substack{v_i \nsim v_j}}x_j^2 \right) \nonumber \\
	&=\alpha^2\delta^2\sum_{i=1}^n(n-d_i -1)x_i^2	+(1-\alpha)^2 \delta \sum_{i=1}^n (n-d_i -1)x_i^2 \nonumber\\
	&= \alpha^2\delta^2 (n-1) + (1-\alpha)^2\delta(n-1) - \alpha^2\delta^2 \sum_{i=1}^nd_ix_i^2 - (1-\alpha)^2\delta \sum_{i=1}^n d_ix_i^2 \label{ineq2::UpperBound_radius}
\end{align}

Therefore, substituting (\ref{ineq2::UpperBound_radius}) into (\ref{ineq1::UpperBound_radius}), we have
\begin{align*}
	\lambda_1^2(A_\alpha(G)) &\leq \alpha^2 \sum_{i=1}^n d_i^2 + (1-\alpha)^2\sum_{i=1}^n d_i - \delta(n-1)(\alpha^2\delta+(1-\alpha)^2) + \alpha^2\delta^2 \sum_{i=1}^nd_ix_i^2 + (1-\alpha)^2\delta \sum_{i=1}^n d_ix_i^2	\\
	& \leq \alpha^2 \sum_{i=1}^n d_i^2 + 2m(1-\alpha)^2 + \delta(\alpha^2\delta+(1-\alpha)^2)(\Delta - n + 1)
\end{align*}
From Theorem \ref{theo::sharp_DAS_upper} we have
\begin{align*}
    \lambda_1^2(A_\alpha(G)) \leq \alpha^2 (2mn -n(n-1)\delta + 2m(\delta-1)) + 2m(1-\alpha)^2 + \delta(\alpha^2\delta+(1-\alpha)^2)(\Delta - n + 1).
\end{align*}
	
Then
\begin{align} \label{ineq3::UpperBound_radius}
	\lambda_1(A_\alpha(G)) &\leq \sqrt{\alpha^2 (2mn -n(n-1)\delta + 2m(\delta-1)) + 2m(1-\alpha)^2 + \delta(\alpha^2\delta+(1-\alpha)^2)(\Delta - n + 1)}
\end{align}
and the result follows.
\end{proof}

\begin{theorem} \label{theo::UpperBound2_SpectralRadius}
    Let $G$ be a graph with $n$ vertices, $m$ edges, maximum degree $\Delta$, minimum degree $\delta$ and $\alpha \in [0,1]$. Then
    \begin{equation} \label{eq::UpperBound2}
        \lambda_1(A_\alpha(G)) \leq \dfrac{1}{2} \left( \delta -1 + \sqrt{(\delta - 1)^2 + 4(\alpha \Delta -\alpha(\delta-1)\delta + (1-\alpha)(2m - \delta(n-1))} \right).
    \end{equation}
    Equality holds if and only if $G$ is regular.
\end{theorem}
\begin{proof}
Let $M$ be any matrix associated to a graph $G$ and $S_v(M))$ the sum of the row of $M$ corresponding to the vertex $v$. It is easy to see that 
\begin{align*}
    S_v(A_\alpha^2(G) - (\delta-1)A_\alpha(G)) &= S_v(A_\alpha^2(G)) - (\delta-1)S_v(A_\alpha(G)).
\end{align*}
From \cite{VN17}, we have
\begin{equation*}
    S_v(A_\alpha^2(G)) = \alpha S_v(D^2(G)) + (1-\alpha)S_v(A^2(G)) 
\end{equation*}
and therefore
{\small
\begin{align*}
    S_v(A_\alpha^2(G) - (\delta-1)A_\alpha(G)) &= \alpha S_v(D^2(G)) + (1-\alpha)S_v(A^2(G)) -\alpha(\delta-1)S_v(D(G)) - (1-\alpha)(1-\delta)S_v(A(G))\\
    &= \alpha S_v(D^2(G)) -\alpha(\delta-1)S_v(D(G)) +  (1-\alpha)S_v(A^2(G) - (1-\delta)A(G)).
\end{align*}
}
From Lemma \ref{lemma:: rowsumsBound}, we have
{\small
\begin{align} \label{ineq::rowsum}
   S_v(A_\alpha^2(G) - (\delta-1)A_\alpha(G)) &\leq \alpha S_v(D^2(G))-\alpha(\delta-1)S_v(D(G)) +  (1-\alpha)(2m - \delta(n-1))
\end{align}
}
The inequality (\ref{ineq::rowsum}), holds for every vertex $v \in V(G)$ and for $\alpha \in [0,1]$. From Lemma \ref{lemma::rowsums} we have that
{\small
\begin{align}
    \lambda_1^2(A_\alpha(G)) - (\delta-1)\lambda_1(A_\alpha(G)) & \leq \alpha d^2(v) - \alpha(\delta-1)d(v) + (1-\alpha)(2m - \delta(n-1)) \nonumber\\
    &\leq \alpha \Delta^2 - \alpha(\delta-1)\delta + (1-\alpha)(2m - \delta(n-1)) \label{eq::eqBound}
\end{align}
}
Then, solving the quadratic inequality we obtain
\begin{equation*}
    \lambda_1(A_\alpha(G)) \leq \dfrac{1}{2} \left( \delta -1 + \sqrt{(\delta - 1)^2 + 4(\alpha \Delta -\alpha(\delta-1)\delta + (1-\alpha)(2m - \delta(n-1))} \right).
\end{equation*}

Now, suppose that $G$ is a $r$-regular graph. So  $\delta = \Delta = r$ and $m = \dfrac{nr}{2}$. So, $$\dfrac{1}{2} \left( \delta -1 + \sqrt{(\delta - 1)^2 + 4(\alpha \Delta -\alpha(\delta-1)\delta + (1-\alpha)(2m - \delta(n-1))} \right) = r = \lambda_1(A_\alpha(G)).$$ Now, suppose the equality holds. Then, all inequalities in the above argument must be equalities. From Lemma \ref{lemma:: rowsumsBound}

$$\sum_{u \nsim v, u \neq v}d(u) = (n-d(v)-1)\delta,$$

for all $v \in V(G)$. Hence either $d(v) = n-1$ or $d(u)=\delta$, for all $u \in V(G)$ and $u \nsim v$, which implies that either $G$ is a regular graph or $G$ is a bidegreed graph in which each vertex is of degree either $\delta$ or $n-1$. As, the second one can not occurs because of inequality (\ref{eq::eqBound}) we get the result.
\end{proof}

\begin{example}
Let $G \cong K_{1,6}$  as in Figure \ref{fig::star}. The Table \ref{tab:table2} presents a comparison of the upper bounds (\ref{eq::UpperBound1}) and (\ref{eq::UpperBound2}) and concludes that they are incomparable.
\begin{figure}[h]%
    \centering    \includegraphics[width=0.3\textwidth]{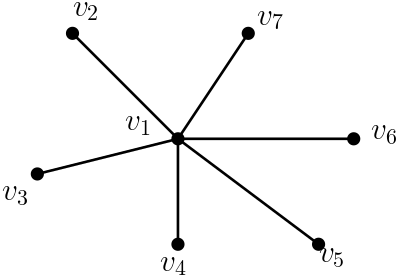}
    \centering
    \caption{Graph $G$.}
    \label{fig::star}
\end{figure}
\begin{table}[!h]
    \centering
    \resizebox{16cm}{!}
    {
    \begin{tabular}{|c|c|c|c|c|c|c|c|c|c|c|}
    \hline
    & $\alpha = 0.0$ & $\alpha = 0.1$ & $\alpha = 0.2$ & $\alpha = 0.3$ & $\alpha = 0.4$ & $\alpha = 0.5$ & $\alpha = 0.6$ & $\alpha = 0.7$ & $\alpha = 0.8$ & $\alpha = 0.9$ \\  
    \hline
    $\lambda_1(A_\alpha(G))$ & 2.44949 &  2.56867 & 2.72237 & 2.9215 & 3.17764 & 3.5 & 3.89165 & 4.34803 & 4.85913 & 5.41329\\
    \hline
    (\ref{eq::UpperBound1}) & 3.4641 & 3.18434 & 3.05941 & 3.10805 & 3.32265 & 3.67423 & 4.12795 & 4.65403 & 5.23068 & 5.84294\\ 
    \hline
    (\ref{eq::UpperBound2}) & 2.44949 & 3.0 & 3.4641 & 3.87298 & 4.24264 & 4.58258 & 4.89898 & 5.19615 & 5.47723 & 5.74456\\
    \hline 
    %(\ref{eq::UpperBound3}) & 2.44949& 2.81522 & %3.19735 &  3.59064 & 3.99181 & 4.39869 & %4.80984 & 5.22424 & 5.64119 & 6.06015\\
    %\hline 
    \end{tabular}
    }
    \caption{Table with upper bounds of $\lambda_1(A_\alpha(G))$ with different $\alpha$ values.}
    \label{tab:table2}
\end{table}
\end{example}

\begin{rmk}
    From Theorem \ref{theo::UpperBound2_SpectralRadius} we have that the equality in (\ref{eq::UpperBound2}) occurs if and only if $G$ is regular. Taking this extremal graph and applying its information in (\ref{eq::UpperBound1}) we obtain
    {\small
    \begin{equation*}
        \alpha^2 (2mn -n(n-1)\delta + 2m(\delta-1)) + 2m(1-\alpha)^2 + \delta(\alpha^2\delta+(1-\alpha)^2)(\Delta - n + 1)  = \alpha^2(r^3 + 2r^2 +r)+ (1 - 2\alpha)(r^2 +r) \geq r^2.
    \end{equation*}}
    Then, the bounds obtained in (\ref{eq::UpperBound1}) is bigger than the bound obtained in (\ref{eq::UpperBound2}) for regular graphs.
\end{rmk}

\subsection{Bounds for \texorpdfstring{$\mathbf{\lambda_2(A_\alpha(G))}$}{lambda2}} \label{subsection::lambda2}

\begin{proposition}
    Let $G$ be a $r$-regular graph with $n$ vertices and independence number $\gamma(G)$ and $\alpha \in [0,1]$. Then,
    \begin{equation}
        \lambda_2(A_\alpha(G)) \geq \alpha r + (1-\alpha)\left(-1 + \dfrac{2(n-1-r)^{\gamma(G)}}{\gamma(G) n^{\gamma(G) - 1}}\right).
    \end{equation}
    Equality holds if $G \cong K_n$.
\end{proposition}
\begin{proof}
    From Theorem \ref{theo::weyl} we know that
    $$\lambda_2(A_\alpha(G)) = \lambda_2(\alpha D(G) + (1-\alpha)A(G)) \geq \alpha \lambda_n(D(G)) + (1-\alpha)\lambda_2(A(G)) = \alpha r + (1-\alpha)\lambda_2(A(G))$$
    and from Theorem \ref{theo::LowerBound_SecondLargest_adjacency} the result follows.
    
    Now, suppose that $G \cong K_n$. Then $\gamma(G) = 1,$ $r = n-1$ and consequently  $$\alpha r + (1-\alpha)\left(-1 + \dfrac{2(n-1-r)^{\gamma(G)}}{\gamma(G) n^{\gamma(G) - 1}}\right) =  \alpha n -1$$ and from Proposition \ref{prop::complete_graph_spectrum} the equality holds.
\end{proof}

\begin{proposition} \label{theo::ineq_bipartite}
	Let $G$ be a connected bipartite graph with $n$ vertices, $\widetilde{G}$ the bipartite complement of $G$, $K_{p,q}$ ($p\leq q)$ the complete bipartite graph whose partitions are the same as those of $G$, and $\alpha \in [0,1]$. Then,
\begin{equation*}
    p\alpha + \lambda_{j+1}(A_\alpha(\widetilde{G})) \leq \lambda_j(A_\alpha(G)) \leq p\alpha + \lambda_{j-1}(A_\alpha(\widetilde{G}))
\end{equation*}
for $ 2 \leq j \leq n-1$.
\end{proposition}
\begin{proof}
It is easy to see that $A_\alpha(G) +A_\alpha(\widetilde{G}) = A_\alpha(K_{p,q})$. From Theorem \ref{theo::weyl} follows that
\begin{equation} \label{eq1::theo_ineq_bipartite}
	\lambda_j(A_\alpha(G)) + \lambda_{i-j+n}(A_\alpha(\widetilde{G})) \leq \lambda_i(A_\alpha(K_{p,q})), \ \ 1 \leq i \leq j \leq n
\end{equation}
and
\begin{equation}\label{eq2::theo_ineq_bipartite}
    \lambda_i(A_\alpha(K_{p,q})) \leq \lambda_j(A_\alpha(G)) + \lambda_{i-j+1}(A_\alpha(\widetilde{G})),\ \ 1 \leq j \leq i \leq n.
\end{equation}
From inequalities (\ref{eq1::theo_ineq_bipartite}) and (\ref{eq2::theo_ineq_bipartite}) we have 
\begin{equation} \label{eq3::theo_ineq_bipartite}
    \lambda_j(A_\alpha(G)) \leq \lambda_i(A_\alpha(K_{p,q})) - \lambda_{i-j+n}(A_\alpha(\widetilde{G}))
\end{equation}
and
\begin{equation} \label{eq4::theo_ineq_bipartite}
    \lambda_i(A_\alpha(K_{p,q})) - \lambda_{i-j+1}(A_\alpha(\widetilde{G})) \leq \lambda_{j}(A_\alpha(G))
\end{equation}
From Proposition \ref{prop::complete_bipartite_spectrum} we obtain $\lambda_2(A_\alpha(K_{p,q})) = \alpha p$ and taking $i=2$ in the inequalities (\ref{eq3::theo_ineq_bipartite}) and (\ref{eq4::theo_ineq_bipartite}) we have
\begin{equation*}
    \lambda_j(A_\alpha(G)) \leq p\alpha - \lambda_{2-j+n}(A_\alpha(\widetilde{G})).
\end{equation*}
and
\begin{equation*}
   p\alpha - \lambda_{3-j}(A_\alpha(\widetilde{G})) \leq \lambda_j(A_\alpha(G))
   \end{equation*}

Moreover, from Theorem \ref{theo::bipartite_spectrum} follows that $\lambda_{2-j+n}(A_\alpha(\widetilde{G})) = -\lambda_{j-1}(A_\alpha(\widetilde{G}))$ and $\lambda_{3-j}(A_\alpha(\widetilde{G})) = - \lambda_{j+1}(A_\alpha(\widetilde{G}))$. Then

\begin{equation*}
    p\alpha  + \lambda_{j+1}(A_\alpha(\widetilde{G}))  \leq \lambda_j(A_\alpha(G)) \leq p\alpha + \lambda_{j-1}(A_\alpha(\widetilde{G})),\ \  2 \leq j \leq n-1.
\end{equation*}
\end{proof}

\begin{corollary} \label{cor::ineq_bipartite}
Let $G$ be a connected bipartite and $r$-regular graph with $n$ vertices. If  $\widetilde{G}$ is the bipartite complement of $G$ and $\alpha \in [0,1]$ then
\begin{equation} \label{eq::secondlargest_upper1}
    \lambda_2(A_\alpha(G)) \leq \dfrac{n}{2}(\alpha + 1) -r.
\end{equation}
\end{corollary}
\begin{proof}
Let $G$ be a connected bipartite and $r$-regular graph such that $V = V_1 \cup V_2$, where $\vert V_1 \vert = p$ and $\vert V_2 \vert = q$ ($p \leq q$). From Proposition \ref{theo::ineq_bipartite} we knows that
\begin{equation*}
	\lambda_j(A_\alpha(G)) \leq p\alpha + \lambda_{j-1}(A_\alpha(\widetilde{G})),\ \  2 \leq j \leq n-1.
\end{equation*}
Taking $j=2$ we have
\begin{equation*}
	\lambda_2(A_\alpha(G)) \leq p\alpha + \lambda_{1}(A_\alpha(\widetilde{G}))
\end{equation*}
As $G$ is $r$-regular bipartite, it follows that $p = \dfrac{n}{2}$. Furthermore, we know that the graph $\widetilde{G}$ is also bipartite and $\left(\displaystyle \frac{n}{2}-r\right)$-regular which implies $\lambda_{1}(A_\alpha(\widetilde{G})) = \dfrac{n}{2} - r $, and consequently the result follows.
\end{proof}

The upper bound obtained by Corollary \ref{cor::ineq_bipartite} can be improved, as we can see in Proposition \ref{theo::adapt_alpha_matrix_bound_lambda2}.  

\begin{proposition} \label{theo::adapt_alpha_matrix_bound_lambda2}
Let $G$ be a $r$-regular bipartite and connected graph with $n$ vertices. Then
\begin{equation}\label{eq::secondlargest_upper2}
    \displaystyle \lambda_2(A_\alpha(G)) \leq \alpha\left(2r - \frac{n}{2}\right) + \frac{n}{2} - r.
\end{equation}

\end{proposition}
\begin{proof}
Since $G$ is $r$-regular and bipartite from Theorems \ref{theo::bipartite_spectrum} and Proposition \ref{prop::regular_graph_spectrum} follows that $\lambda_1(A(G)) = r$ and $\lambda_n(A(G)) = -r$ are eigenvalues of $A(G)$. Moreover, we have that $\widetilde{G}$ is $\left(\displaystyle \frac{n}{2}-r\right)$-regular and bipartite, so $\displaystyle \lambda_1(A(\widetilde{G})) = \frac{n}{2}-r$ and $\displaystyle \lambda_n(A(\widetilde{G})) = -\frac{n}{2}+r$ are eigenvalues of $A(\widetilde{G})$. From Theorem \ref{theo::equality_roots_polynomial} and Lemma \ref{lemma::eigeneq_RegularGraphs} we have $\lambda_k(A(G)) = \lambda_k(A(\widetilde{G})), \text{ for } 2 \leq k \leq n-1$,

\begin{equation} \label{eq1::eigeneq_RegularBip}
\lambda_k(A_\alpha(G)) = \alpha r + (1-\alpha)\lambda_k(A(G)), \text{ for } 1 \leq k \leq n
\end{equation}
and
\begin{equation} \label{eq2::eigeneq_RegularBipComp}
\lambda_k(A_\alpha(\widetilde{G})) = \alpha \left(\displaystyle \frac{n}{2}-r\right) + (1-\alpha)\lambda_k(A(\widetilde{G})), \text{ for } 1 \leq k \leq n. 
\end{equation}

So, $\lambda_1(A_\alpha(G)) = r$,  $\lambda_1(A_\alpha(\widetilde{G})) = \displaystyle \frac{n}{2} - r$,  $\lambda_n(A_\alpha(G)) = r(2\alpha - 1)$ and $\lambda_n(A_\alpha(\widetilde{G})) = \left(\displaystyle \frac{n}{2} - r\right)(2\alpha - 1)$.

Subtracting the Equation (\ref{eq2::eigeneq_RegularBipComp}) from (\ref{eq1::eigeneq_RegularBip}) and using the relation between $\lambda_k(A_\alpha(G))$ and $\lambda_k(A_\alpha(\widetilde{G}))$ for $2 \leq k \leq n-1$ we get
\begin{equation}
	\lambda_k(A_\alpha(G)) = \alpha\left(2r - \displaystyle \frac{n}{2}\right) + \lambda_k(A_\alpha(\widetilde{G}))
\end{equation}
Taking $k = 2$ and considering $\lambda_1(A_\alpha(\widetilde{G})) \geq \lambda_2(A_\alpha(\widetilde{G}))$, we have that
\begin{equation}
	\lambda_2(A_\alpha(G)) \leq \alpha\left(2r - \displaystyle \frac{n}{2}\right) + \frac{n}{2} - r = \frac{n}{2}(1-\alpha) + r(2\alpha - 1)
\end{equation}
and the result follows.
\end{proof}

\begin{rmk}
    It is worth noting that if $G$ is a regular connected and bipartite graph and the regularity is $\dfrac{n}{2}$ then the upper bounds, shown in the equations (\ref{eq::secondlargest_upper1}) and (\ref{eq::secondlargest_upper2}), are equal.
\end{rmk}

\begin{proposition}
	Let $G$ be an $r$-regular graph of order $n$. Then $\lambda_1(A_\alpha(G)) + \lambda_2(A_\alpha(G)) \leq 2r\alpha + (1- \alpha)(n-2)$. Equality holds if, and only if, $G$ has a connected component that is a bipartite graph.
\end{proposition}
\begin{proof}
Let $G$ be a $r$-regular graph of order $n$. From Theorem \ref{theo::LargestEigenvaluesSum} and Lemma \ref{lemma::eigeneq_RegularGraphs} we have
\begin{equation} 
  \label{eq1::eigeneq_regular}
	\lambda_1(A_\alpha(G)) + \lambda_2(A_\alpha(G)) = 2\alpha r + (1-\alpha)(\lambda_1(A(G)) + \lambda_2(A(G))) \leq 2\alpha r + (1-\alpha)(n-2)
\end{equation}
with equality occurs if and only if $G$ has a connected component that is a bipartite graph and the result follows.
\end{proof}

\begin{proposition}
	Let $G$ be a $r$-regular bipartite graph of order $n$ and $\alpha \in [0,1)$. Then $\lambda_1(A_\alpha(G)) + \lambda_2(A_\alpha(G)) \leq \displaystyle 2r\alpha + (1-\alpha)\frac{n}{2}$. Equality occurs if, and only if, $\widetilde{G}$ is disconnected.
\end{proposition}
\begin{proof}
Let $G$ be a $r$-regular bipartite graph. From Lemma \ref{lemma::eigeneq_RegularGraphs} we have $\lambda_1(A_\alpha(G)) = r$.
Moreover, from Proposition \ref{theo::adapt_alpha_matrix_bound_lambda2} we have
$$\displaystyle \lambda_1(A_\alpha(G)) + \lambda_2(A_\alpha(G)) \leq r + \alpha\left(2r - \frac{n}{2}\right) + \frac{n}{2} - r = \alpha\left(2r - \frac{n}{2}\right) + \frac{n}{2} = 2r\alpha + (1- \alpha)\frac{n}{2}.$$
Now, suppose that $\displaystyle \lambda_1(A_\alpha(G)) + \lambda_2(A_\alpha(G)) = 2r\alpha + (1- \alpha)\frac{n}{2}$. As $G$ is $r$-regular, we have that
$$ \lambda_1(A_\alpha(G)) = \alpha r + (1-\alpha)\lambda_1(A(G)) \text{ and } \lambda_2(A_\alpha(G)) = \alpha r + (1-\alpha)\lambda_2(A(G))$$
and consequently
$$\displaystyle \lambda_1(A_\alpha(G)) + \lambda_2(A_\alpha(G))= 2\alpha r + (1-\alpha)(\lambda_1(A(G)) + \lambda_2(A(G))).$$
So,
$$\displaystyle \lambda_1(A(G)) + \lambda_2(A(G)) = \frac{n}{2}$$
and from Corollary \ref{cor::KS2013_corollary31} follows that $\widetilde{G}$ is disconnected.

Now, suppose that the $\widetilde{G}$ is disconnected. From Corollary \ref{cor::KS2013_corollary31},  $\displaystyle \lambda_1(A(G)) + \lambda_2(A(G)) = \frac{n}{2}.$  As $G$ is $r$-regular we have
$$\displaystyle \lambda_1(A_\alpha(G)) + \lambda_2(A_\alpha(G))= 2\alpha r + (1-\alpha)(\lambda_1(A(G)) + \lambda_2(A(G))) = 2\alpha r + (1-\alpha)\frac{n}{2}.$$
and the result follows.
\end{proof}

\subsection{Bounds for \texorpdfstring{$\mathbf{\lambda_n(A_\alpha(G))}$}{lambdan}} \label{subsection::lambdan}

\begin{proposition}
Let $G$ be a graph with n vertices, minimum degree $\delta \neq 0$, maximum degree $\Delta$ and independence number $\gamma(G)$. Then
\begin{equation} \label{eq::BoundLambda_n}
    \lambda_n(A_\alpha(G)) \leq \alpha \Delta + (1-\alpha) \dfrac{\gamma(G) \delta^2}{\lambda_1(A(G))(\gamma(G) - n)}
\end{equation}
In particular, if $G$ is an $r$-regular graph we have
\begin{equation} \label{eq::BoundLambda_nRegular}
\lambda_n(A_\alpha(G)) \leq \alpha r + (1-\alpha) \dfrac{\gamma(G) r}{\gamma(G) - n},
\end{equation}
whose equality holds if $G \cong K_n$ or,  when $n$  is even, $G \cong \displaystyle \bigcup_{i=1}^{\frac{n}{2}}K_2$, or $G \cong K_{\frac{n}{2}} \cup K_{\frac{n}{2}}$.
\end{proposition}

\begin{proof}
From Corollary \ref{cor::weyl} and Theorem \ref{lambda_nGeneralBound} we obtain the bound in (\ref{eq::BoundLambda_n}) and from Corollary \ref{cor::weyl} and Theorem \ref{theo::HoffmanBound} we have the bound in (\ref{eq::BoundLambda_nRegular}). Initially suppose that $G \cong K_n$. Then $\gamma(G) = 1$ and from Proposition \ref{prop::complete_graph_spectrum}, $\lambda_n(A_\alpha(G)) = \alpha n -1$. So, $\alpha r + (1-\alpha) \dfrac{\gamma(G) r}{\gamma(G) - n} = \alpha n - 1.$

Now suppose that $n$ is even and $G \cong \displaystyle \bigcup_{i=1}^{\frac{n}{2}}K_2$. We know that $\gamma(G) = \dfrac{n}{2}$, $r = 1$, $\lambda_n(A_\alpha(G)) = 2 \alpha - 1$ and consequently $\alpha r + (1-\alpha) \dfrac{\gamma(G) r}{\gamma(G) - n} = 2\alpha - 1$. Finally, suppose that $n$ is even and  $G \cong K_{\frac{n}{2}} \cup K_{\frac{n}{2}}$. In this case, $\gamma(G) = 2$, $r = \dfrac{n}{2} - 1$ and $\lambda_n(A_\alpha(G)) = \dfrac{n}{2} \alpha - 1$. So $\alpha r + (1-\alpha) \dfrac{\gamma(G) r}{\gamma(G) - n} = \alpha \dfrac{n}{2} - 1$, and the result follows.
\end{proof}

\begin{proposition}
    Let $G$ be a triangle-free graph with n vertices and independence number $\gamma(G)$. Then
    \begin{equation*}
        \lambda_n(A_\alpha(G)) \leq \alpha \Delta + (1-\alpha) \dfrac{\lambda_1^2(A(G))}{\lambda_1(A(G)) - n}
    \end{equation*}
\end{proposition}
\begin{proof}
From Theorem \ref{theo::weyl} we have that $\lambda_n(A_\alpha(G)) \leq \alpha \Delta + (1-\alpha) \lambda_n(A(G))$ and applying  Theorem \ref{lemma::lambda_nBoundTriangleFree}  the result follows.
\end{proof}

\section*{Acknowledgments}

The research of C. S. Oliveira is supported by CNPq Grant 304548/2020-0.

%Bibliography
\bibliographystyle{unsrt}  
\bibliography{bounds}

\end{document}